\documentclass{article}

\usepackage{PRIMEarxiv}

\usepackage[utf8]{inputenc} % allow utf-8 input
\usepackage[T1]{fontenc}    % use 8-bit T1 fonts
\usepackage{hyperref}       % hyperlinks
\usepackage{url}            % simple URL typesetting
\usepackage{booktabs}       % professional-quality tables
\usepackage{amsfonts}       % blackboard math symbols
\usepackage{nicefrac}       % compact symbols for 1/2, etc.
\usepackage{microtype}      % microtypography
\usepackage{lipsum}
\usepackage{fancyhdr}       % header
\usepackage{graphicx}       % graphics

\usepackage[export]{adjustbox}

\usepackage[english]{babel}

\usepackage{doi}

\usepackage{comment}  % to use comment

\usepackage{amsmath} % math 

\usepackage{amssymb} % math symbol

\usepackage{amsthm}

\usepackage[ruled,linesnumbered]{algorithm2e}%ruled, vlined

\usepackage{float}

\usepackage{subcaption}

\usepackage{silence}

\usepackage{dirtytalk}

\newtheorem{theorem}{Theorem}[section]
\newtheorem{lemma}[theorem]{Lemma}
\newtheorem{definition}[theorem]{definition}

\graphicspath{{media/}}     % organize your images and other figures under media/ folder

\title{ an efficient validated asynchronous byzantine agreement protocol using committee}
\author{ {Nasit S Sony} \\
	University of California, Merced\\
	CA 95340, USA \\
	\texttt{nsony@ucmerced.edu} \\
}

\begin{document}
\maketitle

\begin{abstract}
We present a Byzantine agreement protocol to address the inefficiencies inherent in multi-valued Byzantine agreement protocols, i.e., a version of the Byzantine agreement protocol where every party broadcasts its request, and at the end of the protocol, every party agrees on one of the party's requests. The protocol we present is a validated asynchronous Byzantine agreement protocol, i.e., a party's request must be validated by some external validity property before it is proposed for agreement. Differently from most of the MVBA protocols, we allow only a subset of total parties to broadcast their requests instead of all, and we make the subset selection stochastic each time the parties choose to broadcast a new set of requests. Then, at the time of the agreement, we choose a party from the selected subset, and the parties reach an agreement on the selected party's broadcast. Extensive theoretical analysis shows that this approach can produce efficient output regarding messages and computation overhead, but the protocol is time-consuming.

\end{abstract}

% keywords can be removed

\section{Introduction}

Byzantine agreement (BA) is a fundamental problem in computer systems that can be used to design fault-tolerant systems, ensure data consistency, design consensus protocols for blockchain and cryptocurrency, decentralize decision-making, and ensure security against malicious actors, to name a few. In its basic formulation, BA assumes a system where multiple computers try to agree on a value, given that some of the computers might be Byzantine (whose behavior is arbitrary, unpredictable, and malicious). Many models with numerous system assumptions have been presented since the seminal work \cite{BYZ23} to solve the BA problem. It has been shown that the parties are unable to reach an agreement in an asynchronous network even if there is only one non-Byzantine fault. Most of the practical solutions of the protocol assumed either a synchronous or partially synchronous model. Since Bitcoin \cite{BITCOIN01}, the BA problem has received renewed interest, and given the demand, synchronous and partially synchronous models no longer provide the practical solution. Miller et al. \cite{HONEYBADGER01} prove that a protocol in a partially synchronous model cannot make progress in \textit{an intermittently synchronous network}. Therefore, interest has moved towards a solution for asynchronous networks. But most asynchronous BA protocols are theoretical and inefficient \cite{ SECURE02,SECURE03,SECURE06}. VABA \cite{BYZ17} is an $O(n^2)$ algorithm with the running time $O(1)$ in expectation where parties agree on one party's requests. The standard concept in the asynchronous network is letting every party broadcast the request because it does not use any leader to track the faults. Letting every party broadcast its request is beneficial at the time of atomic broadcast protocol design \cite{HONEYBADGER01, FASTERDUMBO} because, in the atomic broadcast protocol, the parties agree on multiple parties' requests and mitigate the overhead of high communication complexity. However, when parties agree on one party's output, letting every party broadcast its requests is a waste of resources. Though there is a study with fewer broadcasts, those received far less attention and do not guarantee a solution with probability \cite{KSizeVABA}. 

%Miller et al. \cite{HONEYBADGER01} provides an atomic broadcast protocol that outputs multiple parties' requests at a time, which is followed by \cite{FASTERDUMBO}. The atomic broadcast protocol lets the parties agree on multiple parties' requests and incurs a high communication complexity. But if different parties broadcast the requests and have duplication, then allowing every party proposes each time is not beneficial in terms of message and communication complexity.

 %Byzantine agreement (BA) is a fundamental problem in computer systems introduced by Lamport, Pease, and Shostak in their pioneering works \cite{BYZ22,BYZ23}. BA assumes a system where multiple computers try to agree on a value, but some of the computers might be Byzantine (whose behavior is arbitrary, unpredictable, and malicious). We will use machine, party, and node synonymously. Many models with numerous system assumptions have been presented since the seminal work \cite{BYZ23} to solve the BA problem.

 In our recent works \cite{PMVBA,OHBBFT,SlimABC,cMVBA,KSizeVABA,sony2024agreement}, we presented a new approach to solve the multi-valued Byzantine agreement problem. We introduce a committee approach to provide a more efficient protocol. However, the introduction of committee brought new challenges when we tried to utilize the concept in the VABA protocol to enjoy efficiency. The VABA protocol \cite{BYZ17} is a view-based Byzantine agreement protocol. In each view, parties either decide on a value or adopt a value to promote in the next view. The protocol employs a proposal-promotion sub-protocol consisting of four sequential provable-broadcast steps. This structure closely resembles a three-step lock-commit protocol, with an additional fourth step to generate a commit proof. The commit proof ensures that at least $f+1$ honest parties have received a \textit{commit} for the leader's proposal. If the leader successfully completes the commit proof, the parties decide on the leader's proposed value, as $n-2f$ honest parties will have received the commit and disseminated it during the view-change phase.To increase the likelihood of agreement ($\frac{2}{3}$ probability), a party waits for $n-f$ commit proofs before electing a leader. If the elected leader fails to complete the broadcast, all parties adopt the leader's proposed value as a key and promote it in the next view. This ensures convergence, as every party promotes the same key, leading to agreement in subsequent views. Therefore, we take the question: \textit{Does allowing a subset of parties to broadcast their requests provide any improvement over the existing protocols?}

 To answer the question, we find out the associated challenges. After finding the challenges, we have done the following analysis and propose a solution based on the analysis.
 
 %we design a protocol that improves the efficiency of the existing protocols. Our two main observations are that we can reduce the number of proposals, and the reduction helps to design an efficient protocol. We apply our reduction technique to design a validated asynchronous byzantine agreement (VABA) protocol. 

%To address the above-mentioned challenges, 

\subsection{Challenges and Proposed Solutions}
To design a protocol that reduces the number of messages and computational overhead, we address the following key challenges:
\begin{enumerate}
    \item \textbf{Determining Necessary Broadcasts:} How many broadcasts are essential to maintain protocol progress?
    \item \textbf{Selection of Broadcasting Parties:} How can we select a subset of parties for broadcasting?
    \item \textbf{Preventing Dishonest Broadcasts:} How can we prevent non-selected dishonest parties from broadcasting?
    \item \textbf{Dispersing Broadcasts:} How can the broadcast messages efficiently reach all parties?
    \item \textbf{Adapting Leader Election:} How can we adapt the leader election process to ensure leaders are chosen from the selected subset?
\end{enumerate}

\paragraph{} To address the first and second challenges, we analyze adversarial (network and Byzantine parties). A dishonest party or adversary can read and delay messages from honest parties but must eventually deliver them to ensure protocol progress. Since the protocol requires at least one honest broadcast to proceed—and it's impossible to determine if a single party is honest—we select a set of parties large enough to include at least one honest party with overwhelming probability. We utilize the standard committee selection protocol to realize this property. The protocol randomly selects a subset of parties, with the subset of size-$f+1$. By choosing parties, we ensure that at least one honest party is included with overwhelming probability.

\paragraph{} To prevent the dishonest parties from making progress with their proposals, we enhance the provable-broadcast protocol. In the standard provable-broadcast protocol, every party can broadcast its request. This openness allows a non-selected dishonest party to bypass selection criteria and promote its proposal. To mitigate this, we introduce a security check in the provable-broadcast protocol. Honest parties only respond to proposals from selected parties, preventing dishonest parties from completing the first step of the proposal-promotion sub-protocol.

\paragraph{} Since the proposed protocol deviates from the standard approach of proposal selection, we need to ensure both the dispersal of messages and efficiency given the change condition. In asynchronous protocols, the standard practice is to wait for $n-f$ broadcasts to ensure progress, given that $f$ parties may be faulty. However, if only one honest party broadcasts, other parties must respond upon receiving the first broadcast to maintain liveness. We adapt this behavior to align with the reduced broadcasting subset.

\paragraph{} The VABA protocol uses the standard leader election protocol to elect a party for the agreement. However, the standard leader election protocol is insufficient as it may elect any party, including non-selected ones, as a leader. Therefore, we must adapt the \textit{Leader Election} for selected parties. To address this, we combine the leader election protocol with a mapping mechanism that ensures the elected leader belongs to the selected committee. This adaptation guarantees that the leader has been chosen from the subset of the parties, maintaining the integrity of the protocol.

By addressing these challenges, our proposed enhancements to the VABA protocol significantly reduce communication and computation overhead while maintaining the resilience and the standard properties of the VABA protocol.

The protocol's functionality can be understood as follows: Although only $f+1$ parties broadcast their requests, the protocol guarantees that at least one party successfully completes the promotion phase. The suggestion step ensures that a completed broadcast is disseminated to multiple parties. An honest party decides on a proposal from a party $p_i$ if the leader election protocol selects $p_i$ as the leader and the honest party has received either a commit proof or a commit message from $p_i$. The VABA protocol ensures consensus among all honest parties, such that they decide on the same proposal. Additionally, if the elected leader completes either the second or third step of the proposal-promotion phase, at least $f+1$ honest parties will have received the corresponding key proof or lock proof. This guarantees that the value is adopted and promoted in the next view, ensuring protocol progress and eventual agreement.

\subsection{Our Contribution}
We design a protocol that improves the efficiency of the existing VABA protocol. Our two main observations are that we can reduce the number of proposals, and the reduction helps to design an efficient protocol. We apply our reduction technique to design a validated asynchronous byzantine agreement (VABA) protocol.

\begin{itemize}
    \item Integrate the committee selection protocol that ensures the inclusion of at least one honest party with probability 1, ensuring the protocol's termination property.
    \item Propose eVABA, reducing communication messages by using the technique of number of proposal broadcast reduction.
    \item Integrate a leader election protocol and introduce the mapping of the leader to a selected party to achieve an agreement on a proposal by the selected parties.
    \item Provide theoretical analysis showing the efficiency and the correctness of the protocol.
    % \item Validate cMVBA’s efficiency and scalability through case studies, demonstrating superior performance over traditional MVBA protocols.
\end{itemize}

\section{System Model}
We assume an asynchronous message passing system \cite{BYZ17,FASTERDUMBO, HONEYBADGER01}, which consists of a fixed set of parties. If the total number of parties is $n$ and the maximum number of byzantine parties is $f$, then $f <  \frac{n}{3} $.

\paragraph{Computation.} The model uses standard modern cryptographic assumptions and definitions from \cite{SECURE02, SECURE03}. We prototype the system modules’ computations as probabilistic Turing machines and provide $infeasible$ problems to the adversary thus the adversary is unable to solve the problem. We define a problem as an $infeasible$ if any polynomial time probabilistic algorithm solves it only with negligible probability. Since the computation modules are probabilistic Turing machines, the adversary solves a problem using a probabilistic polynomial time algorithm.  But from the definition of $infeasible$ problem, the probability to solve at least one such problem out of a polynomial in $k$ number of problems is negligible. So, we bound the total number of parties $n$ by a polynomial in $k$.

\paragraph{Communications.} We consider an asynchronous network, where communication is point-to-point, and the medium is reliable and authenticated \cite{BYZ20,BYZ30}. Reliability ensures that if an honest party sends a message to another honest party, then the adversary can determine the time of the delivery of the messages but unable to read, drop or modify the messages. On the other hand, an authenticated
medium ensures that if a party $p_i$ receives a message $m$, then another party $p_j$ sends the message $m$ before the party $p_i$ receives the message.
\subsection{Design Goal}
\subsubsection{Validated asynchronous byzantine agreement (VABA)} Our goal is to design a byzantine agreement protocol in asynchronous network where parties agree on a value proposed by some party. We follow Abraham et al. \cite{BYZ17} and define Efficient-VABA protocol that adopted the notion of the VABA with a selected valid function. This function ensures a proposed value is acceptable for a particular application.
\begin{definition}[Validated byzantine agreement]
A protocol solves the validated byzantine agreement with the base $externally-valid$ function, if the protocol satisfies the following conditions except with negligible probability:

\begin{itemize}
    \item \textbf{External validity:} Any honest party that terminates and decides $v$ such that $externally-valid\langle v \rangle = true$.
    \item \textbf{Agreement:} If an honest party decides $v$, then all the honest parties that terminate decide $v$.
    \item \textbf{Liveness:} If all honest parties participate and all associated messages get delivered, then all honest parties decide.
    \item \textbf{Integrity:} If the honest parties follow the protocol and decide on a value $v$, then $v$ is proposed by some party $valid\langle v, p_i \rangle$ = true, where value $v$ is proposed in the current or the previous view.
    \item \textbf{Efficiency:}  The number of message generated by the honest parties is probabilistically uniformly bounded.
\end{itemize}
\end{definition}

\section{Preliminaries}
\paragraph{Provable-Broadcast.} Provable-Broadcast for the selected parties satisfies except with negligible probability:
\begin{itemize}
    \item \textbf{PB-Integrity:} An honest party delivers a message at most once.
    \item \textbf{PB-Validity:} If an honest party $p_i$ delivers $m$, then $EX-PB-VAL_i\langle id,m \rangle = true.$
    \item \textbf{PB-Abandon-ability:} An honest party does not deliver any message after it invokes PB-abandon(ID).
    \item \textbf{PB-Provability:} For two $v$, $v'$, if a sender can produce tow threshold-signature $\sigma$, $\sigma'$ s.t. threshold-validate($ \langle id, v \rangle, \sigma$) = true, then threshold-validate($\langle id, v'\rangle, v'$) = true, i) $v=v'$ and ii) f+1 honest parties delivered a message m s.t. m.v = v.
    \item \textbf{PB-Termination:} If the sender is honest, no honest party invokes PB-abandon(ID), all messages among honest parties arrive, and the message m that is being broadcast is externally valid, then (i) all honest parties deliver m, and (ii) PB(ID, m) return (to the sender) $\sigma$, which satisfies $threshold-validate(\langle ID, m.v \rangle,\sigma)$ = true.
    \item \textbf{PB-Selected:} If an honest party $p_i$ delivers $m$, then $m$ is proposed by a selected party.
   
\end{itemize}
\paragraph{Threshold signature scheme:} We adopt the threshold-signature scheme from \cite{SECURE02,THRESH01,THRESH02}. In $\langle n, t, f\rangle$ threshold-signature scheme, there is a total of $n$ parties, among them $f$ can be faulty and $t$ sign-shares are both necessary and sufficient to create a threshold-signature, where $f<t \leq n$. The signature scheme satisfies the following properties except with negligible probability:

\begin{itemize}
    \item \textit{Unforgeability:} No polynomial-time adversary can forge a signature that can be verified correctly (by honest parties) of any message $m$ without querying the signature algorithm.
    \item \textit{Robustness:} If a message gets a signature from the signature algorithm, eventually all honest parties get the signature can verify the signature.
\end{itemize}
\paragraph{Threshold coin-tossing:} We assume a trusted third party as a dealer which provides a pseudo-random generator (PRG)$ G: R \rightarrow \{1,...,n\}^s$ that gets a string as input and returns a set size of $s$.

At the beginning of the protocol, the dealer provides a private function $CoinShare_i$ to every party $P_i$, and two public functions: $Coin-Share-Verify$ and $Coin-Toss$. The function $Coin-Toss$ takes $f+1$ $coin-share$ as input and returns a unique and pseudorandom set. The scheme satisfies the following properties with negligible probability:

\begin{itemize}
    \item For each party $i \in \{1..n\}$ and for every string $r$, $Coin-Share-Verify\langle r, i, \sigma \rangle = true$ if and only if $\sigma = CoinShare_i(r)$;
    \item An adversary is unable to create a share $CoinShare_i(r)$ for an honest party $P_i$.
    \item For every string r, $Coin-Toss\langle r, \sum\rangle = G(r)$ if and only if $\sum \geq f+1$ and $\sigma \in \sum$ and a party $P_i$ s.t. $Coin-Share-Verify\langle r,i,\sigma \rangle = true$.
\end{itemize}

\paragraph{\textbf{$(1, \kappa, \epsilon)$}- Committee Selection}
A CS protocol is executed among $n$ parties (identified from 1 through $n$). If at least $f$ + 1 honest parties participate, the protocol terminates with honest parties output a $\kappa$-sized committee set $C$ such that at least one of $C$ is an honest party. %The detailed properties and the pseudocode is given in Appendix \ref{CS}.

\paragraph{Committee Selection Properties.}
The protocol satisfies the following properties except with negligible probability in the cryptographic security parameter:$f+1$:

\begin{itemize}
    \item \textbf{Termination.} If $\langle f+1 \rangle$ honest parties participate in committee selection and the adversary delivers the messages, then honest parties output $C$. 
    \item \textbf{Agreement.} Any two honest parties output same set $C$. 
    \item \textbf{Validity.} If any honest party outputs $C$, then (i) $|C| = (f+1)$, (ii) The probability of every party $p_i \in C$ is same, and (iii) $C$ contains at least one honest party with probability $1-\epsilon$.
    \item \textbf{Unpredictability.} The probability of the adversary to predict the returned committee before an honest party participates is at most $\frac{1}{^nC_(f+1)}$.
\end{itemize}

\paragraph{Proof for Validity property.} Algorithm \ref{algo:cs} satisfies the Termination, Agreement and Unpredictability properties of CE follows from the properties of threshold coin-tossing \cite{FASTERDUMBO}. We only repeat the validity proof here.

\begin{lemma}(Validity of Committee Selection.) If $n=3f+1$, $f+1 \leq f$ and CE(id) returns a set $C$ containing at least one honest party except with $exp(-\epsilon(\kappa))$ probability.
\end{lemma}
%The proof is adopted from \cite{FASTERDUMBO}.
\begin{proof}
    Due to pseudo-randomness, hence, the total case is $\binom{n}{\kappa}$ of random choose $\kappa$ parties and the total case is $\binom{f}{\kappa}$ of the set C containing no honest parties. Therefore we have, \newline
    $p =  \frac{\binom{f}{\kappa}}{ \binom{n}{\kappa}} = \frac{ \frac{f!}{\kappa!(f-\kappa)!}}{ \frac{n!}{\kappa!(n-\kappa)!}} \leq (\frac{1}{3})^\kappa = exp(-\epsilon(\kappa))$
\end{proof}

\section{The Proposed Protocols}
\begin{figure}[h]
    \centering
    \includegraphics[width=0.882\textwidth,height=0.4\textwidth]{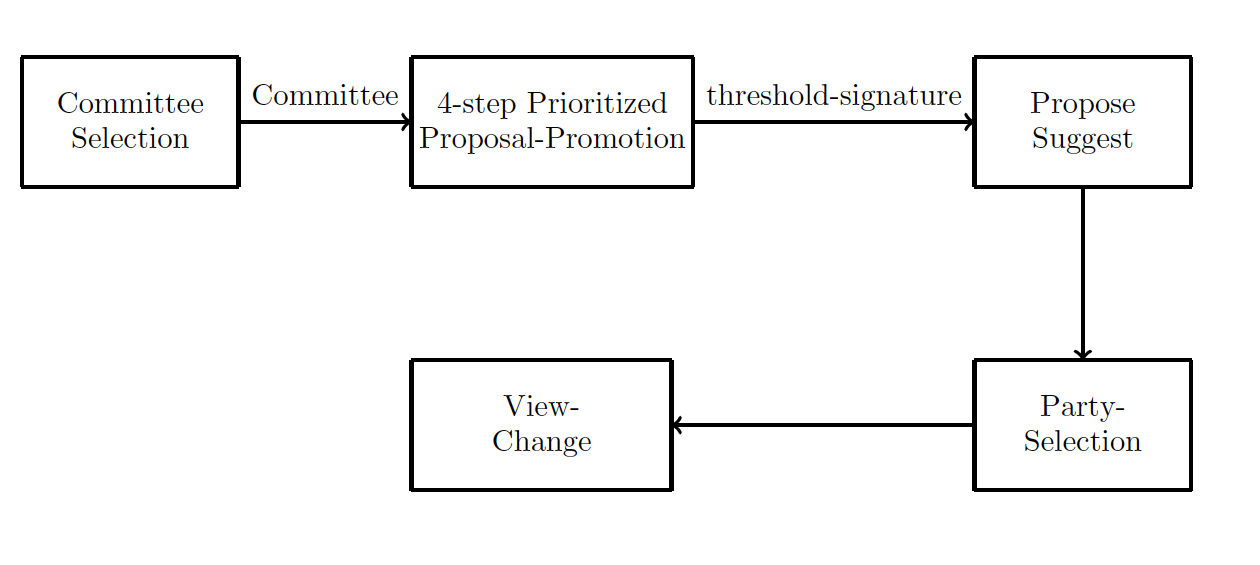}
    \caption{An overview of the proposed protocol.}
    \label{fig:E-VABA}
\end{figure}
In this section, we present the proposed protocol. We first provide the task of each step of the protocol and finally, the integration of the steps. Since we propose an algorithm that improves the VABA protocol, we adopted the definitions and the protocols from VABA \cite{BYZ17}. The visual description of the Efficient-VABA protocol's steps is given in Figure \ref{fig:E-VABA}

\subsection{Committee Selection} Our main contribution is to use a committee to improve the efficiency of the VABA \cite{BYZ17} protocol. Therefore, we need to choose the size of the committee and the selection process in a way; hence, there is at least one honest party to ensure the progress of the protocol. We use the standard committee selection protocol, where we choose $(f+1)$ as a security parameter bounded with $n$, the party selection is random, and at least one selected party is honest.
\subsubsection{Committee selection protocol.} The committee selection protocol is given in Algorithm \ref{algo:cs}. We use the standard cryptographic abstraction to realize the randomness in the time committee selection. The committee selection is standard practice and adopted from pMVBA \cite{PMVBA}.
\begin{algorithm}[hbt!]
%\SetAlgoLined
\LinesNumbered
\DontPrintSemicolon
\SetAlgoNoEnd
\SetAlgoNoLine

\SetKwProg{LV}{Local variables initialization:}{}{}
\LV{}{
   $\Sigma \leftarrow \{\}$\;
}

\SetKwProg{un}{upon}{ do}{}
\un{$SelectCommittee \langle id, view, (f+1)\rangle $ invocation}
{
  $\sigma _i$ $\leftarrow$ $CShare\langle id \rangle$ \;
  $multi-cast \langle SHARE, id, \sigma _i, view \rangle$\;
  wait until $|\Sigma| = f+1$\;
  
  \KwRet $CToss \langle id, \Sigma, (f+1)) \rangle$\;
}
\SetKwProg{un}{upon receiving}{ do}{}
\un{$ \langle SHARE, id, \sigma _k, view \rangle$ from a party $p_{k}$ for the first time}
{
 \uIf{$CVerify\langle id_{k}, \sigma _k\rangle = true$ }{
    $\Sigma \leftarrow {\sigma _k \cup \Sigma}$}
}

\caption{Committee - Selection: Protocol for party  $p_i$}
\label{algo:cs}
\end{algorithm}
\subsection{Prioritized Proposal-Promotion}
%Though we select $(f+1)$ number of parties as a promoter, a non-selected Byzantine party may try to promote its proposal to occupy the network traffic. So, we are unable to use the same $provable-broadcast$ sub-protocol to promote the requests. To prevent a Byzantine party from broadcasting its requests, we modified the $provable-broadcast$ sub-protocol, where a party only replies to a broadcast if the broadcast is from a selected party. We name the changed $provable-broadcast$ as Prioritized-provable broadcast (P-PB). Therefore, 
The prioritized proposal-promotion protocol is adapted from the VABA\cite{BYZ17} protocol. The protocol invokes the prioritized provable-broadcast protocol four times. Each time the protocol returns a threshold-signature which is the input of the next step. The final step's return value is proof that the protocol has completed the promotion. The construction of prioritized provable-broadcast is given in subsection \ref{P-PB-D}.
\subparagraph*{Construction of Prioritized Proposal-Promotion (sender):} The pseudocode of the protocol is given in Algorithm \ref{alg:Proposal-Promotion}. The description of the protocol is given below: 

\begin{itemize}
    \item Upon the invocation of a promote protocol, a party invokes the P-PB protocol four times sequentially. The arguments of the protocol are the step number, value, and proof ($t_{sign}$). Finally, the party returns the $t_{sign}$ to the caller. (lines 02-05)
\end{itemize}

\begin{algorithm}
%\SetAlgoLined
\AlgoDontDisplayBlockMarkers
\SetAlgoNoEnd
\SetAlgoNoLine
\DontPrintSemicolon
\SetKwProg{un}{upon}{ do}{}
\un{Promote $\langle ID, \langle v, prepare\rangle \rangle$  invocation}
{
  $t_{sign} \leftarrow \bot$\;
  \For{$step\leftarrow 1$ \KwTo $4$}{
    $t_{sign} \leftarrow P-PB \langle \langle ID, step\rangle, \langle v,\langle prepare, t_{sign}\rangle \rangle \rangle$\;
   
  } \KwRet $t_{sign}$\;
}
\caption{Proposal-Promotion with identification ID: Protocol for a sender.}
\label{alg:Proposal-Promotion}
\end{algorithm}

\subparagraph*{Construction of Prioritized Proposal-Promotion (receiver):} The pseudocode of the prioritized proposal-promotion (receiver) protocol is given in Algorithm \ref{algo:PP(receiver)}. The description of the protocol is given below:

\begin{itemize}
    \item Since the prioritized proposal-promotion invokes the P-PB protocol four times, and input of the next step is the returned threshold-signature of the previous step. So, when a party receives the delivery for P-PB, it gets a proof $t_{sign}$ of completion of the previous step (except the first step). The party saves the proof and value it receives at each step. (line 03-09)
    \item If a party receives an abandon invocation, the party abandons the P-PB protocol's messages for all steps. (line 10-12).
    \item An honest party returns the proof for each step upon invocation. (line 13-18)
\end{itemize}

\begin{algorithm}
%\SetAlgoLined
%\AlgoDontDisplayBlockMarkers
\SetAlgoNoEnd
\SetAlgoNoLine
\DontPrintSemicolon
\SetKwProg{LV}{Local variables initialization:}{}{}
\LV{}{
 \textbf{prepare = lock = commit =$\langle \bot, \bot \rangle$} \;
}
   \BlankLine
   \BlankLine
\SetKwProg{un}{upon}{ do}{}
\un{$delivery\langle \langle ID,step\rangle, \langle v, \langle prepare, t_{in}\rangle\rangle\rangle$ }
{
  \If{$step = 2$}{
    $ {prepare} \leftarrow \langle v, t_{in}\rangle$\;
   }
  \If{$step = 3$}{
     $ {lock} \leftarrow \langle v, t_{in}\rangle$\;
  }
  \If{$step = 4$}{
     $ {commit} \leftarrow \langle v, t_{in}\rangle$\;
  }
  
}

\BlankLine
\SetKwProg{un}{upon}{ do}{}
\un{$abandon\langle ID\rangle$ }
{
   \For{$step=1,..,4$}{
      $P-PB-abandon\langle ID,step \rangle$\;
  }
}
\BlankLine
\BlankLine
\BlankLine
\SetKwProg{un}{upon}{}{}
\un{$getPrepare\langle ID\rangle$}{
   \KwRet prepare\;
}

\un{$getLock\langle ID\rangle$}{
    \KwRet lock\;
}
\un{$getCommit\langle ID\rangle$}{
    \KwRet commit\;
}

\caption{Proposal-Promotion with identification ID: Protocol for a receiver $p_i$.}
\label{algo:PP(receiver)}
\end{algorithm}

\subsubsection{Prioritized probable broadcast (P-PB)\label{P-PB-D}} 
We adopt the provable-broadcast from VABA to ensure that the adopted protocol only allows the committee members to broadcast their requests. This approach saves the channel from sending additional messages and cryptographic computation.  Provable-Broadcast for the selected parties can be obtained through same threshold-signature scheme. Generally, in the provable-broadcast, a party adds sign-share to a request if the request is authentic or sent by some party. But, we design a protocol where a party adds sign-share to a request only if the request is from a selected party for the particular view. Therefore, only the selected parties receive $\langle n-f \rangle$ replies with $sign-share$ and produce a $threshold-signature$ on a message. 

\subparagraph*{Construction of Prioritized-PB:} The pseudocode of the P-PB protocol is given in Algorithm \ref{PB:P-PB1}. We provide a flow of the protocol below: 

\begin{itemize}
    \item Upon the invocation of a P-PB protocol, a party  creates a message type $SEND$ using the $ID$ and the received value and proof $\langle v, \sigma \rangle$. Then the party multi-casts/broadcasts the created message. (lines 03-06)
    \item Upon receiving a $SEND$ type message from a party $p_k$, a party checks whether the sender is a selected party for the particular $view$. Then the party checks the value and proof. The checking depends on the view and the step of the P-PB (see algorithm \ref{algo: PB-messages}).
    \item  Upon receiving a $sign-share$ $\sigma_k$ from a party $p_k$, a party checks the authenticity of the message. If the sender is authentic, the party adds the $sign-share$ $\rho$ to its set $\Sigma$. (line 08-09)
    \item A selected sender waits for $\langle n-f \rangle$ valid $sign-shares$ and upon receiving the valid $sign-shares$, the party returns the $threshold-signature$ of the $sign-shares$. (lines 05-06)
\end{itemize}

\begin{algorithm}[hbt!]
%\SetAlgoLined
\LinesNumbered
\DontPrintSemicolon
\SetAlgoNoEnd
\SetAlgoNoLine

\SetKwProg{LV}{Local variables initialization:}{}{}
\LV{}{
   $\Sigma \leftarrow \{\}$\;
}

\SetKwProg{un}{upon}{ do}{}
\un{$P-PB \langle ID,\langle v, \sigma \rangle\rangle $ invocation}
{
  $multi-cast\langle ID, SEND, \langle v, \sigma \rangle \rangle$ \;
  %$msg$ $\leftarrow$ $\langle id, requests \rangle$ \;
  %$threshold-sign\langle msg \rangle$\;
  %$multi-cast \langle msg \rangle$\;
  \textbf{wait} until $|\Sigma| = n-f$\;
  
  \KwRet $threshold-sign\langle \Sigma \rangle$\; %$CombineShares \langle id, \Sigma \rangle$\;
}
\SetKwProg{un}{upon }{ do}{}

\un{receiving $\langle ID, ACK, \sigma_k \rangle$ from a party $p_{k}$ for the first time}
{
 \uIf{$share-validate (\langle ID, ACK \rangle, k, \sigma_k ) = true$ }{
    $\Sigma \leftarrow {\sigma_k \cup \Sigma}$}
}

\caption{P-PB: Protocol for party  $p_i$}\label{PB:P-PB1}
\end{algorithm}

\subparagraph*{Construction of Prioritized-PB (messages):} The pseudocode of the Provable-Broadcast (messages) is given in Algorithm \ref{algo: PB-messages}. The description of the protocol is given below:

\begin{itemize}
    \item Before each invocation, a party sets its stop value to false (it is waiting for messages). (line 01-02)
    \item Upon receiving a valid ($EX-PB-VAL$) message from a selected party for the first time($stop=false$), the party sets $stop=true$, $sign-share$ the ID and value, delivers the received proof, and replies to the sender with $sign-share$. (line 03-08)
    \item Upon receiving P-PB-abandon, a party sets its $stop=true$ and does not accept any P-PB messages. (line 09-08)
    \item In $EX-PB-VAL$, a party validates value v and the threshold-signature based on the step number and returns $true$ or $false$. (line 11-19)
    \item In check-prepare(v, prepare), a party first checks whether the value is externally valid. A party returns true when the value is externally-valid, $t_{sign}$ valid and the view is greater than $LOCK$. (line 20-31) 
\end{itemize}
\begin{algorithm}
\SetAlgoNoEnd
\SetAlgoNoLine
\DontPrintSemicolon

 \SetKwProg{LV}{Local variables initialization:}{}{}
\LV{}{
$ stop \leftarrow false $\;
}

\BlankLine
\BlankLine

\SetKwProg{un}{upon}{ do}{}
\un{receiving $\langle ID, SEND, \langle v, t_{sign} \rangle \rangle$ from the sender}{
  \If{$ ID.id \in PrioritizedParties \wedge stop = false \wedge EX-PB-VAL(ID, \langle v, t_{sign}\rangle)=true$}{
   $stop \leftarrow true$\;
   $\sigma_i \leftarrow share-sign_i(\langle ID,v\rangle)$\;
   $deliver\langle v, t_{sign}\rangle$\;
   $reply\langle ID, \sigma_i \rangle$\;
  }
}
\BlankLine
\BlankLine

\un{$P-P$B-abandon(ID)}{
   $stop \leftarrow true$\;
}
\BlankLine
\BlankLine

%\SetKwProg{un}{upon}{ do}{}
\SetKwProg{un}{procedure}{}{}

\un{EX-PB-VAL(ID, $\langle v, t_{sign}\rangle$)}{

    parse ID as $\langle \overline{ID}, step \rangle$\;

    \BlankLine
    \BlankLine
    
    parse $t_{sign}$ as $\langle prepare, ts_{in}\rangle$\;

   \If{step=1}{
     \If{check-prepare(v,prepare)}{
        \KwRet true\;
     }
   }
    
    \If{step>1 $\wedge $ threshold-validate($\langle\langle \overline{ID}, step-1 \rangle, v\rangle, t_{sign}$)}{
      \KwRet true\;
    }
    \KwRet false\;
}
\BlankLine
\un{$check-key(v, prepare)$ }
{
\If{EX-VABA-VAl(v) = false}{
    \KwRet false\;
}
  $party \leftarrow Party[view]$\;
  parse $prepare$ as $\langle v, \sigma \rangle$\;
  \If{$view \neq 1$}{
     \If{$threshold-validate \langle \langle \langle \langle id, Party[view], view\rangle,1\rangle, v\rangle,\sigma \rangle =false$}{\KwRet false\;}
  }
  \uIf{$view \geq LOCK$}{
    \KwRet true\;
  }\uElse{
   \KwRet false\;
  }
}
\caption{Provable-Broadcast with identification ID: messages.} \label{algo: PB-messages}
\end{algorithm}

\subsection{Propose-Suggest} The purpose of the Propose-Suggest step is to ensure that parties receive enough proof even though we reduce the number of broadcast. 
Since $(f+1)$ parties broadcast their requests and among them, one party can complete the proposal-promotion (since only one party can be honest), a party is required to act after receiving the first proposal (unlike VABA, where a party waits for $\langle n-f \rangle$ proposals). But a party also needs to know that at least $\langle f+1 \rangle$ honest parties receive at least one proposal before moving on to the proposal-selection. So, we introduce another step, $suggestion$, where each party suggests the received proposal and waits for $\langle 2f+1 \rangle$ suggestions (it indicates that at least $\langle f+1 \rangle$ honest parties have received a proposal). We also require to stop unnecessary message processing (If a party has already received $\langle 2f+1 \rangle$ suggestions, then the party does not require any promotion). The step is given in Algorithm \ref{algo:Efficient-VABA}.

\subsection{Party Selection} \label{proposal-selection}
The purpose of the party selection step is to choose a party from the committee members so all the parties can agree on the selected party's proposal. Therefore, after completing the Propose-Suggest step, a party moves to the next step, where honest parties coordinate to elect a leader. It is required to elect a leader from the committee members, but the standard leader election protocol elects any party among $n$ parties. Therefore, to elect a party from the committee members, our protocol first elects a leader using the leader election protocol, which can be any party from $\{1,2,..,n\}$. If the elected leader is a committee member, a party can deliver that party's deliverable. But if the elected leader is not a committee member, then we utilize a mapping function to map the elected leader to a selected party, and the mapping ensures that the same committee member is elected by all parties. The leader election and mapping are given in subsections \ref{lElec} and \ref{mSParty}.

\subsubsection{Leader Election} \label{lElec}

We utilize the standard leader election protocol to elect a leader. The Leader-Election primitive provides one operation to elect a unique party (called a leader) among the parties \cite{BYZ17}. The pseudocode is given in Algorithm \ref{algo:Proposal-selection}.

 \begin{algorithm}[hbt!]
%\SetAlgoLined
\SetKwProg{LV}{Local variables initialization:}{}{}
\LV{}{
$\Sigma \leftarrow \{\}$\;
}

\LinesNumbered

\SetKwProg{Fn}{function}{}{}
\Fn{$elect\langle id, view \rangle$} 
{
 
  $\sigma _i$ $\leftarrow$ $CShare\langle id \rangle$ \;
  $multi-cast \langle SHARE, id, \sigma _i, view \rangle$\;
  wait until $|\Sigma| = f+1$\;
  
  \KwRet $CToss \langle id, \Sigma \rangle$\;

}
\SetKwProg{un}{upon receiving}{ do}{}
\un{$ \langle SHARE, id, \sigma _k, view \rangle$ from a party $p_{k}$ for the first time}
{
 \uIf{$CVerify\langle id_{k}, \sigma _k\rangle = true$ }{
    $\Sigma \leftarrow {\sigma _k \cup \Sigma}$}
}

\caption{Leader-election: Protocol for all parties}
\label{algo:Proposal-selection}
\end{algorithm}

\subsubsection{Mapping to a selected party} \label{mSParty}
Since the elected leader has a unique $id$, we take the distance of the leader from the committee members and pick the nearest selected party (smaller distance). 

\subparagraph*{Construction of mapToParty.} \label{CMapToParty}The pseudocode of the mapToParty protocol is given in algorithm \ref{algo:mapToParty}. The description of the protocol is given below:

\begin{itemize}
    \item Upon the invocation of a mapToParty protocol, a party computes distances with the elected leader from each prioritized party (difference between the leader $id$ and a prioritized party's $id$) and picks the nearest party. (lines 08-08)
    \item Upon computing the nearest party, the party returns the selected party's id. (line 09)
\end{itemize}

\begin{algorithm}[hbt!]
%\SetAlgoLined
%\AlgoDontDisplayBlockMarkers
\SetAlgoNoEnd
\SetAlgoNoLine
\DontPrintSemicolon
\SetKwProg{un}{upon}{ do}{}

%\SetKwProg{un}{upon}{ do}{}
\un{$mapToParty\langle view, id, PrioritizedParties \rangle$ $invocation$} {
   
    $d \leftarrow \infty$\;
    $party \leftarrow \bot$\;

    \For{$p \in PrioritizedParties$ }{
    $dis \leftarrow |id - p_{id}|$\;
     \uIf{$dis< d$}{
       % $d_2 \leftarrow d_1$\;
        $d \leftarrow dis$\;
        $party \leftarrow p_{id}$\;
     }
    }
    \KwRet party\;
}

\caption{Mapping to a selected party: protocol for all parties}
 \label{algo:mapToParty}
\end{algorithm}

\subsection{View-Change}
The purpose of the view-change step is to complete one view and move to the next view. This step takes place when parties complete party-selection/proposal-selection step. Through view-change, parties share their deliverable for the selected party, which ensures that the highest deliverable of the selected party reaches the $f+1$ parties. After the view-change step, parties start the new view with proposal promotion. The pseudocode of the view change step is given in Algorithm \ref{algo:Efficient-VABA}.   

\subsection{Integration of sub-protocols} In this section, we provide construction of the main protocol that combines all the steps of the protocol. We first give the variable initialization. The second construction is the main protocol followed by the message handle of the protocol. 

\paragraph{Global variable initialization} Each party maintains a set of global variables that is accessible from each protocol. Algorithm \ref{algo:localV} provides the list of the global variables. 

\begin{algorithm}
%\SetAlgoLined
\LinesNumbered
\textbf{Gocal variables initialization:}%\newline

$ LOCK \leftarrow 0$\;
$PREPARE \leftarrow \langle 0, v, \bot \rangle$ \;
$DECIDED \leftarrow false$ \;
for every view $\geq$1 initialize: \newline
$ Party[view] \leftarrow \bot$\newline
$ PPdone[view] \leftarrow \{\}$\newline
$ PPSkip[view] \leftarrow \bot$\newline
$skip = false$ \newline
%$Recommendation[view] \leftarrow \bot$\newline
\caption{Global variables initialization: for a party}
\label{algo:localV}
\end{algorithm}
\paragraph{Construction of Efficient VABA}\label{CEfficientVAAB} 
\begin{itemize}
    \item \textit{Committee-selection.} Every party computes the selected parties for the view using the function SelectCommittee (see Algorithm \ref{algo:cs}) protocol.  (line 04)
    \item \textit{Proposal-promotion.} If a party finds itself in the prioritizedParties, the party generates ID for the current view. If the view is greater than one and the last view has not been completed, then the party adopts the last view's prepare and proof. Otherwise the party adds new requests in the PREPARE's value. (line 06 - 12)
    \item \textit{Proposal-promotion.} Then the party promotes its requests using the $promote$ protocol (see Algorithm \ref{alg:Proposal-Promotion}). This protocol invokes P-PB sub-protocol four times and finally returns a $threshold-signature$. (line 13)
    \item \textit{Proposal-suggestion.} Upon return from the promote, a selected party multi-casts  the $proposal$ and $t_{sign}$. The party also multi-casts the proposal as a suggestion and sets its suggest variable $true$ to indicate that the party is not going to suggest any other proposal. (line 17-20 )
     \item \textit{Proposal-suggestion.} If a non selected party has not suggested any proposal yet and has received a proposal or suggestion,the party multi-casts a suggestion. After sending a suggestion, a party waits for  $\langle n-f \rangle$ suggestions. (line 21-25)
    \item \textit{Proposal-selection.} When a party receives $\langle n-f \rangle$ suggestions, the party multi-casts a $DONE$ message and waits for $skip=true$. A party handles a $DONE$ message similar to the VABA (see Algorithm \ref{alg:E-VABA(messages)}). (line 30 - 33) 
    \item \textit{Proposal-selection.} Upon the $skip= true$, a party abandons all the messages related to promotion. Then the party invokes $elect\langle id, view \rangle$ (see Algorithm \ref{algo:Proposal-selection}) that elects a party and returns the elected party's id. If the elected party does not belong to the committee members, the party maps the elected party to a committee member (see Algorithm \ref{algo:mapToParty}). (line 29-34)
    \item \textit{View-change} Each party multi-casts the selected party's delivery via $VIEW-CHANGE$ step. Upon receiving $\langle n-f \rangle$  $VIEW-CHANGE$e messages, a party decides or adopts a $prepare$ $\langle 0, v, t_{sign} \rangle$ and moves to the next view. (line 35-37)
\end{itemize}

\begin{algorithm}[hbt!]
\DontPrintSemicolon
\SetAlgoNoEnd
\SetAlgoNoLine
\SetKwProg{un}{upon}{ do}{}

 % \SetKwProg{un}{upon}{ do}{}
  
$view \leftarrow 1$ \;
%$Decisson[view] \leftarrow \bot$ \;
\While{true}{
   $suggest \leftarrow false$\;
    $PrioritizedParties \leftarrow SelectCommittee\langle id, view \rangle$\;
   %invoke $promote\langle ID, view, key\rangle$\;
   \If{$id_{p_i} \in PrioritizedParties$}{
       $ID \leftarrow {\langle view, id_{p_i}\rangle}$\;
       \uIf{view>1 and DECIDED = false }{
           $prepare \leftarrow \langle PREPARE.view, PREPARE.proof\rangle$\;
          
       }\Else{
          $PREPARE.view = view$\;
          $PREPARE.value = requests$\;
          %$PREPARE.proof = Decision[view-1]$\;
          $prepare \leftarrow \langle PREPARE.view, PREPARE.proof\rangle$\;
          
       }
           \textbf{invoke} $promote\langle ID,\langle PREPARE.value, prepare\rangle\rangle$\;
      \;
      \;
     
    }
  $ DECIDED \leftarrow false$ \;
   \un{return from promote}{
      multi-cast $\langle proposal, t_{sign} \rangle$ \;
      
      multi-cast $\langle suggestion, t_{sign} \rangle$ \; 
      suggest = true\;
   %   $multi-cast\langle Done \rangle$ \;
     % \textbf{wait until} skip=true\;
   }

   \SetKwProg{un}{upon receiving}{ do}{}

   \un{ a proposal or suggestion }{
        
     \If{suggest = false}{
        multi-cast $\langle suggestion, t_{sign} \rangle$ \;
        $suggest \leftarrow true$ \;
    }
   }
      
    \un{ $\langle n-f \rangle$ suggestions}{
         \If{skip[view] = false}{
            $multi-cast\langle id, DONE, view, PREPARE.value, t_{sign} \rangle$ \;
         }  
  }
    
 \textbf{wait until} skip=true\;

 \SetKwProg{un}{upon}{ do}{}
   \un{skip = true}{
    \For{$id \in PrioritizedParties$ }{
          $abandon\langle id, view\rangle$\;
   }

    $electedParty \leftarrow elect (id, view)$\;
     \If{$electedParty \notin PrioritizedParties$}{
        $electedParty \leftarrow mapToParty(view, electedParty)$\;
     }

   %$elect$ leader\;
   $multi-cast\langle id, VIEW-CHANGE, view, getPrepare\langle ID\rangle, getLock\langle ID\rangle, getCommit\langle ID\rangle$\;
   \textbf{wait} for $\langle n-f \rangle$ VIEW-CHANGE messages\;
   %decide or adopt key for the elected leader\;
   $view \leftarrow view+1$\;
    }

}
\caption{Efficient-VABA: protocol for party $p_i$}
\label{algo:Efficient-VABA}
\end{algorithm}

\paragraph*{Construction of Efficient-VABA (messages):} The proposed Efficient-VABA protocol deals the messages similar to the base protocol. The pseudocode of the Efficient-VABA (messages) is given in Algorithm \ref{alg:E-VABA(messages)}. The description of the protocol is given below:

\begin{itemize}
    \item  Upon receiving a $DONE$ message, a party validates the threshold signature of the accompanying proposal (the proposal has completed all four steps) and increments the PPdone[view]. If the party receives $\langle n-f \rangle$ $DONE$ messages, then the party generates a sign-share of the $SKIP$ messages and multi-casts a $SKIP-SHARE$ message. (line 01-07)
    \item When a party receives a valid $SKIP-SHARE$ message, the party adds the $sign-share$ to the set PPskip[view]. The party also sets its skip[view] true (since already $\langle f+1 \rangle$ honest parties have broadcast the $DONE$ messages). Upon receiving $\langle n-f \rangle$ $SKIP-SHARE$ messages, a party threshold-signs on the $SKIP-SHARE$ message and multi-casts the $SKIP$ messages. (line 08-14)
    \item  When a party receives a valid $SKIP$ message, the party sets its skip[view]=true and multi-casts a $SKIP$ message if it does not have done so already. (line 15-19)
    \item Through view-change message, a party delivers the elected party's delivery. If the party has a delivery for the fourth step ($v_4$), the party decide on the value and set the $DECIDED=true$. If the party has delivery for the third step ($v_3$) and the view is greater than the existing lock, then the party changes its $lock$ variable.  If the party has delivery for the second step ($v_2$) and the view is greater than the existing $prepare.view$, then the party changes its $PREPARE$ variable. (line 20-30)
\end{itemize}

\begin{algorithm}
%\SetAlgoLined
%\AlgoDontDisplayBlockMarkers
\SetAlgoNoEnd
\SetAlgoNoLine
\DontPrintSemicolon
\SetKwProg{un}{upon}{ do}{}

\un{$receiving \langle id, DONE, view,v, t_{sign} \rangle$ from party $p_k$ for the first time in view $view$}{
   \If{threshold-validate$\langle \langle \langle \langle id, k, view\rangle, 4 \rangle v \rangle t_{sign}\rangle$}{
   $PPdone[view] \leftarrow PPdone[view] + 1$\;
       \If{ SKIP-SHARE message was not sent yet in view $view$}{
          \If{$|PPdone[view]|=n-f$}{
          $\sigma \leftarrow sign-share \langle id, SKIP, view \rangle$\;
          $multi-cast\langle id, SKIP-SHARE, view, \sigma \rangle$\;
        }
       }  
     }
}

\BlankLine
\BlankLine
\un{$receiving \langle id, SKIP-SHARE, view,\rho \rangle$ from party $p_k$ for the first time in view $view$}{
   \If{share-validate$\langle \langle id, SKIP, view\rangle, k, \rho \rangle$}{
   $PPskip[view] \leftarrow PPskip[view] \cup \{\sigma \}$\;
        $skip[view] \leftarrow true $\;
        \If{$|PPskip[view]|=n-f$}{
          $t_{sign} \leftarrow threshold-sign \langle PPskip[view]\rangle$\;
          $multi-cast\langle id, skip, view, t_{sign}\rangle$\;
        }
     }
}

\BlankLine
\BlankLine
\un{$receiving \langle id, SKIP, view, t_{sign} \rangle$}{
   \If{threshold-validate$\langle \langle id, skip, view\rangle, t_{sign} \rangle$}{
        $skip[view] \leftarrow true $\;
        \If{\textit{SKIP message was not sent yet in view $view$}}{
        $multi-cast\langle id, skip, view, t_{sign} \rangle$\;
        }
     }
}

\SetKwProg{un}{upon}{ do}{}
\un{$receiving \langle id, view-change, view, \langle v_2, t_2\rangle, \langle v_3, t_3\rangle, \langle v_4, t_4 \rangle$ }
{
  $party \leftarrow Party[view]$\;
  
  \If{$v_4 \neq \bot$}{
     \If{threshold-validate$\langle \langle \langle \langle id, party, view\rangle, 3 \rangle, v_4\rangle, t_4)\rangle$}{
        decide $v_4$\;
        $DECIDED \leftarrow true$\;
     }
  }
  \If{$v_3 \neq \bot \wedge view > lock $}{
     \If{threshold-validate($\langle \langle \langle \langle id, party, view\langle,2\rangle,v_3\rangle,t_3\rangle$)}{
        $ LOCK \leftarrow view$\;
     }
  }
  
  \If{$v_2 \neq \bot \wedge view > prepare.view $}{
     \If{threshold-validate$\langle \langle \langle \langle id, party, view\rangle , 1\rangle, v_2\rangle , t_2\rangle)$}{
        $ PREPARE \leftarrow \langle view, v_{2}, t_{2}\rangle$\;
     }
  } 
}

\caption{Efficient-VABA with identification ID: messages.}
\label{alg:E-VABA(messages)}
\end{algorithm}

\section{Security and Efficiency Analysis}
\subsection{Security Analysis}
To analyze the security properties of the protocol, we follow the approach from VABA \cite{BYZ17}. We only state and change the lemma definitions to match the proposed protocol requirements and skip the proofs if it follows the VABA.

\paragraph{Prioritized provable-braodcast} The proposed P-PB protocol satisfies the P-PB-integrity, P-PB-validity, P-PB-abandonability and P-PB linear complexity properties from the code. (similar to provable-broadcast)
\newline 
The P-PB protocol satisfies the termination and provability properties from the provable-broadcast.

\begin{lemma} \label{E-VABA: selected}
    Algorithm \ref{P-PB-D} satisfies the selected property.
\end{lemma}

\begin{proof}
   From the provability property of the protocol, a party generates a threshold-signature as proof for its broadcast. A party requires at least $\langle f+1 \rangle$ honest parties' sign-shares on the promoted message to generate a threshold-signature. From line 4 of Algorithm \ref{algo: PB-messages}, an honest party replies with a sign-share only if the sender is selected for the view. Therefore, only the selected parties can complete the promotion.
\end{proof}

%\paragraph{Efficient-VABA}
\begin{lemma} \label{lem:l4}
    Consider a 4-step P-PB with identification id. For every value $v$ and $j\in\{2,3,4\}$, if some honest party gets a $\sigma$ such that: threshold-validate($\langle\langle\langle id,j\rangle,v\rangle, \sigma\rangle$) = true, then at least f+1 honest parties previously $deliver_j(id, \langle v, \sigma' \rangle)$ for some $\sigma'$.
\end{lemma}

\begin{proof}
    The proof follows from the P-PB provability property (similar to PB-Provability).
\end{proof}
\begin{lemma} \label{lem:l5}
    Consider a 4-step P-PB with identification id. For every value $v$ and $j\in\{2,3,4\}$, if some honest party $deliver_j(id, \langle v, \sigma \rangle)$, then threshold-validate($\langle\langle\langle id,j-1\rangle,v\rangle, \sigma_{in}\rangle$) = true.
\end{lemma}

\begin{lemma}\label{lem:l6}
    Consider a 4-step P-PB with identification id, two honest parties $p_1$, $p_2$, and two values $v_1$, $v_2$. For every $j_1$, $j_2$ $\in$ $\{2,3,4\}$, if $p_1$ gets $\sigma_1$ such that threshold-validate($\langle\langle\langle id,j_1\rangle,v_1\rangle, \sigma_{1}\rangle$) = true and $p_2$ gets $\sigma_2$ such that threshold-validate($\langle\langle\langle id,j_2\rangle,v_2\rangle, \sigma_{2}\rangle$) = true, then $v_1=v_2$.
\end{lemma}

\begin{lemma}\label{lem:l7}
    Consider a 4-step P-PB with identification id. If the sender is honest, no honest parties invoke abandon(id), all messages among honest parties arrive, and the message $m$ that is being broadcast is externally valid, then the broadcast completes and returns a completion proof to the sender.
\end{lemma}

\begin{lemma} \label{lem:l8}
    Consider a view $j$ and let $p_l$ be the chosen leader of view $j$. For every $i$ $\in$ $\{2,3\}$, if an honest party gets a $view-change$ message that include $\langle v, \sigma \rangle$ such that threshold-validate($\langle\langle\langle id,l,j\rangle,i\rangle, v\rangle, \sigma$) = true, then all honest parties get a $view-change$ message that includes $\langle v, \sigma'\rangle$ such that threshold-validate($\langle\langle\langle id,l,j\rangle,i-1\rangle, v\rangle, \sigma'$) = true.
\end{lemma}

\begin{lemma}\label{lem:l9}
   If a party $p$ decides on a value $v$ in a view $j$, then the $LOCK$ the variables of all honest parties are at least $j$ when they move to view $j+1$.
\end{lemma}

\begin{lemma}\label{lem:l11}
   If an honest party decides on $v$ in view $j$, then all honest parties that decide in view $j$, decide $v$ as well.
\end{lemma}

\begin{lemma} \label{lem:l12}
   Assume an honest party $p$ decides on $v$ in a view $j$. Then every view $j\geq j'$, and every honest party that gets $\langle v' \sigma' \rangle$ such that : threshold-validate($\langle\langle\langle id,l',j'\rangle,1\rangle, v'\rangle, \sigma'$) = true, where $p_{l'}$ is the leader of view $j'$, we get that $v'=v$.
\end{lemma}

\begin{lemma} \label{lem:l14}
Consider a view $j$ and assume that some honest party $p$ sets its $LOCK$ variable to $j$ in view $j$, then all honest parties set their $PREPARE$ variable to $\langle j, v, \sigma \rangle$ such that: $PREPARE = \langle j, \sigma \rangle$ is valid for $v$ before moving to view $j+1$.
\end{lemma}

\begin{lemma} \label{lem:l15} 
If all selected honest parties start with externally valid values for byzantine agreement, then for every view $j\geq 1$, all messages broadcast by honest parties are externally valid for the 4-step P-PB.
\end{lemma}

\begin{proof}
In view $j = 1$, since all selected parties start with externally valid values for the Byzantine agreement, but a valid prepare is not required in view $1$, we see that all messages broadcast by the selected honest parties in view 1 are externally valid for the 4-step P-PB.

In view $j > 1$. Let i = max({l | there is an honest party $p$ who’s $lock = l$ when it begins view $j$}). We need to show that when a party
begins view $j$, its $PREPARE$ variable is equal to $\langle l, v, \sigma_i \rangle$ such that: 

\begin{enumerate}
    \item $\langle l, v, \sigma_i \rangle$ is a valid prepare for $v$.
    \item $l \geq i$.
    \item $v$ is an external valid value for byzantine agreement
\end{enumerate}

From lemma \ref{lem:l14}, all honest parties set their $PREPARE$ variable to $\langle i, v, \sigma \rangle$ when they completes view $i$ such that: $prepare = \langle j, \sigma \rangle$ is valid for $v$. After view $i$, the $PREPARE$ variable is updated to $\langle l, v$'$
, \sigma$'$\rangle$ by an honest party $p$ only if
the party $p$ receives a $VIEW-CHANGE$ message with $\langle v', \langle l,\sigma $ $'$$\rangle\rangle$  such that:  $l > i$ and $\langle l, \sigma'$$\rangle$ is a valid $prepare$ for $v'$. Therefore, we see that (1) and (2) are satisfied. Now by the P-PB-provability property, we get that at least one honest party delivers $\langle v, \sigma''$$\rangle$ for some $\sigma''$ at some instance of P-PB, and so $v$ is an external valid value for byzantine agreement.
\end{proof}

\begin{lemma} \label{lem:l17}
Assume that all selected honest parties start with externally valid values for a byzantine agreement. If all messages sent by honest parties in view $j$ arrived and no honest party sets $skip=true$, then at least one 4-step P-PB issued by an honest party has completed and returned completion-proofs.
\end{lemma}

\begin{proof} From lemma \ref{lem:l15}, all messages broadcast by the selected honest parties are externally valid for the 4-step P-PB protocol. Since no honest party sets $skip = true$ in view $j$, no honest party invokes abandon in
a the 4-step P-PB protocol of view $j$. Therefore, from lemma \ref{lem:l7}, all broadcast issued by the selected honest parties have completed and returned completion-proofs.
\end{proof}

\begin{lemma} \label{lem:l18}
   Assume a view $j$, if all messages sent among honest parties in a view $j$ arrived and a 4-step P-PB issued by an honest party in view $j$ has been completed and returned completion-proofs, then all honest parties set $skip = true$ in view $j$.
\end{lemma} 

\begin{proof}
From the algorithm, an honest party broadcasts a $DONE$ message with a completion-proof (of its proposal or received proposal) upon receiving $\langle n-f \rangle$ suggestions. Therefore, since all messages sent among honest parties in a view $j$ arrived, we get that honest parties receive s$\langle n-f \rangle$ suggestions and all
honest parties broadcasts a $SKIP-SHARE$ message. There are at least $\langle n-f \rangle$ honest parties, all honest parties generated a valid threshold-signature on $SKIP$ and multi-casts the threshold-signature. Thus, all honest
parties gets a $SKIP$ message with a valid threshold-signature and set $skip = true$ in view $j$.
\end{proof}

\begin{lemma}  \label{lem:l19}
Assume the selected honest parties start with externally valid values for byzantine agreement, if all messages sent among honest parties in a view $j$ arrived, then view $j$ completed.
\end{lemma}

\begin{proof}
We first show that all honest parties set $skip = true$ in view $j$. we also observe that if all messages sent by honest parties in a view $j$ arrived and some honest party sets $skip = true$ in view $j$, then all honest parties set $skip = true$ in that view. Let, if no honest party sets $skip = true$ in view $j$. Then from lemma \ref{lem:l17}, at least one 4-step P-PB issued by a selected honest party has completed, proposed and suggested. Therefore, from lemma \ref{lem:l18}, all honest parties sets $skip = true$ in view $j$. A contradiction.

We conclude that no honest party waits forever for elect(j) to return or to a $VIEW-CHANGE$ message from $\langle n - f \rangle$ different parties. Since, all honest parties set $skip = true$ in view $j$, we see from code that they all invoke elect(j), and thus by the termination property of the leader election abstraction we find that all invocations of elect(j) by honest parties return. Therefore, all honest parties multi-cast a $VIEW-CHANGE$ message and all messages sent among honest parties in a view $j$ arrived, we get that all honest parties get $VIEW-CHANGE$ messages from at least $\langle n - f \rangle$ different parties.
Hence, all parties move to view $j + 1$.
\end{proof}

\begin{theorem}
   The Efficient-VABA protocol satisfies the Agreement, External-Validity, and Liveness/Termination properties of the VABA protocol, given that the underlying Committee selection, P-PB and Leader election protocols are secure. 
\end{theorem}

\begin{proof}
\textit{Agreement:} Let a party $p$ first decides on a value $v$ in view $j$. From lemma \ref{lem:l11}, all honest parties that decide in view $j$, decide $v$. Now, if we consider a view $j'>j$, from lemma \ref{lem:l12}, if an honest party gets $\langle v', \sigma' \rangle$ such that: threshold-validate$(\langle \langle \langle id, l', j'\rangle, 1\rangle, v'\rangle, \sigma') = true$, where $p_l'$ is the leader of view $j'$, then $v=v'$. Thus by lemma \ref{lem:l6}, if an honest party gets $\langle v', \sigma'\rangle$ such that: threshold-validate$(\langle \langle \langle id, l', j'\rangle, 3\rangle, v'\rangle, \sigma') = true$, where $p_l'$ is the leader of view $j'$, then $v=v'$. Therefore, since an honest party decides on $v'$ in view $j'$ only if it gets the value $v'$ as a $view-change$ message with $\langle v' \sigma' \rangle$ such that: threshold-validate$(\langle \langle \langle id, l', j'\rangle, 3\rangle, v'\rangle, \sigma') = true$, we get that if an honest party decide $v'$ in view $j'$, then $v=v'$.

\textit{ Termination:}  From lemma \ref{lem:l19}, we know every view completes. After applying the lemma \ref{lem:l19} inductively we notice that the number of views is unbounded and the honest parties keep sending messages. Therefore, the protocol satisfies the termination property. 

\textit{External-Validity:} To prove the External-Validity property of the Efficient-VABA, we prove that if the selected honest parties start with external valid value $v$, then all messages broadcast by honest parties are externally valid for the P-PB protocol. 

Lemma \ref{E-VABA: selected} proves that only the selected parties can broadcast their requests. The rest of the proof follows the lemma \ref{lem:l15}. In view $j = 1$, since all selected parties start with externally valid values for byzantine agreement, but, a valid prepare is not required in view $1$, we see that all messages broadcast by the selected honest parties in view 1 are externally valid for the 4-step P-PB.

In view $j > 1$. Let i = max({l | there is an honest party $p$ who’s $lock = l$ when it begins view $j$}). We need to show that when a party
begins view $j$, its $PREPARE$ variable is equal to $\langle l, v, \sigma_i \rangle$ such that: 

\begin{enumerate}
    \item $\langle l, v, \sigma_i \rangle$ is a valid prepare for $v$.
    \item $l \geq i$.
    \item $v$ is an external valid value for byzantine agreement
\end{enumerate}

From lemma \ref{lem:l14}, all honest parties set their $PREPARE$ variable to $\langle i, v, \sigma \rangle$ when they completes view $i$ such that: $prepare = \langle j, \sigma \rangle$ is valid for $v$. After view $i$, the $PREPARE$ variable is updated to $\langle l, v$'$
, \sigma$'$\rangle$ by an honest party $p$ only if
the party $p$ receives a $VIEW-CHANGE$ message with $\langle v', \langle l,\sigma $ $'$$\rangle\rangle$  such that:  $l > i$ and $\langle l, \sigma'$$\rangle$ is a valid $prepare$ for $v'$. Therefore, we see that (1) and (2) are satisfied. Now by the P-PB-provability property, we get that at least one honest party delivers $\langle v, \sigma''$$\rangle$ for some $\sigma''$ at some instance of P-PB, and so $v$ is an external valid value for byzantine agreement.
\end{proof}

\subsection{Efficiency Analysis}

First, let us briefly go through the process of the  Efficient-VABA protocol. From Algorithm \ref{algo:Efficient-VABA}, the message exchange happens when the parties select $f+1$ parties using the committee selection protocol. When the selected parties promote their request using the $proposal-promotion$ sub-protocol, propose and suggest a completed promotion using the $propose-suggest$, elect a leader, and deliver the value and proof using view-change.

 The proposed protocol ensures the termination in the view change step. Since the protocol exits after the constant rounds of messages, the expected \textit{time complexity} of the protocol is $O(1)$. The main message and communication cost of the Efficient-VABA protocol comes from the four P-PB sub-protocol invocations. In the first step of the P-PB sub-protocol, $f+1$ parties broadcast the messages to $n$ parties, and in the second step, $n$ parties reply to $(f+1)$ parties. So the total number of messages is $n*(f+1)$, which reduces the total number of messages and the related computations by a factor of $(f+1)$. The $suggestion$, $leader-election$ and $view-change$ steps exchange three messages, and the communication is all-to-all, so the message complexity is $O(n^2)$ and the communication complexity is $O(n^2(L+K))$. Both are optimal. Here, $L$ is the bit length of the message and the proof; $K$ is the bit length of digital signatures.

\section{Conclusion}
In this paper, we have presented a new algorithm to solve the asynchronous Byzantine agreement problem. The main novelty is using a subset request to reduce the number of messages and computations. To the best of our knowledge, this approach is novel. The proposed algorithm offers various advantages. We reduce the number of provable broadcast instances by choosing a subset of parties as a leader. The reduced number of broadcasts also makes a party wait for only one proof of completion instead of $n-f$, which also significantly reduces the wait time of the protocol. We have also proved that most of the steps of the protocol require only $n\times(f+1)$ number of messages, in contrast to $O(n^2)$ for the VABA protocol. In addition, we have shown that the protocol maintains optimal properties like resilience, message, and communication complexities.

There are various venues for further research. First, the current protocol still requires $n$ to $n$ communications for some of the steps. Future research can discover how to remove the $n$ to $n$ communication. 

%We presented two asynchronous byzantine agreement protocols. The core technical contributions include the introduction of the prioritized provable-broadcast (P-PB) protocol. The P-PB protocol broadcasts only the selected parties' requests and provides an $O(n|v|)$ communication complexity for a value $v$ compared to $O(n|v| + Kn^2logn)$ communication complexity of the RBC protocol. Though the P-PB protocol is unable to provide the totality property of the RBC protocol, we overcome the limitations by adding an extra round of messages.

%We presented Efficient-VABA, an instantiation of VABA \cite{BYZ17} that has reduced the number of messages and the related computations from $n^2$ to $(f+1) n$. We also presented an atomic broadcast protocol, Slim-HBBFT, that provides an instantiation of the HoneyBadgerBFT\cite{HONEYBADGER01} protocol that performs better when parties have duplicate requests. The proposed protocol aimed to agree on the fraction of $n$ parties' requests and removed the $kn^3logn$ term from the communication complexity. 

%\bibliographystyle{ACM-Reference-Format}
\bibliography{references}

\appendix

\section{Related work} \label{Related :Work}
 A Byzantine agreement protocol can reach an agreement if there are a total of $n$ parties, among them $f$ are faulty, and $n$ is at least greater than $3f$ \cite{BYZ10}. Therefore, to have an optimal resilience, a protocol must satisfy the following constraint:$n=3f+1$. Fischer, Lynch, and Paterson\cite{CONS03} gave a theorem that proved that Byzantine agreement protocol does not have a termination property in asynchronous settings even if there is only one non-Byzantine failure. After that, Ben-or \cite{BYZ11} proved that with the help of randomness, these protocols can terminate almost with a probability of $1$. The classic work of Cachin et al. \cite{SECURE02} presented asynchronous binary agreement (ABA), which is the building block of the MVBA protocol \cite{SECURE02}. To realize the security and the randomness, the protocol utilizes the cryptographic abstraction \cite{THRESH01,BORN01}. This abstraction is also used to design the fault-tolerant state machine replication protocol \cite{CACHIN01,VICTOR01,CACHIN02,BYZ20}. Though Cachin's work is first practical solution for an asynchronous network and has optimal resilience, it suffers from the high communication complexity. Recent work of Abraham et al. \cite{BYZ17} removes the $O(n^3)$ term from the communication complexity. Dumbo-MVBA \cite{BYZ20} also uses the MVBA as a base and achieves $O(n^2)$ message complexity but uses erasure code to minimize the message complexity.\; 

Before the surge of DeFi applications, most of the protocol assumed the partially synchronous communication model, which was introduced by Dwork, Lynch, and  Stockmeyer \cite{BYZ05}. This model assumes a known time bound $\Delta$ for message delay; that is, honest parties deliver their messages in this time bound after a \textit{global stabilization time (GST)}. After GST, a protocol advances deterministically \cite{CONS03}. 

Castro et al. \cite{BYZ08} provided the first Byzantine fault-tolerance protocol that assumes a partially synchronous model. The core of the protocol is a leader. If a leader fails to deliver the result in $\Delta$ time-bound, then the parties start the leader election to elect a new leader. An adversary, with the help of the Byzantine parties, can exploit this time-bound to drive the parties to find a new leader and make the leader election process infinite \cite{HONEYBADGER01}. Many other protocols proposed in the literature \cite{BYZ24,BYZ25,BYZ26,BYZ27,BYZ28,BYZ29,Hotstuff01} face the same challenges.

\end{document}